\documentclass[12pt]{article}

\usepackage{soul,tikz, adjustbox}
\usepackage{forest, tikz-qtree,tkz-graph, pseudocode,ulem, framed}
\newsavebox{\tempbox}
\newcommand{\cbox}[2]{%
 \fcolorbox{black}{#1}{\texttt{#2\strut}}\kern-\fboxrule}
\usepackage{times}
\usepackage{palatino}
\usepackage{amssymb,amsmath,amsthm}
\usepackage{epsfig}
\usepackage{xcolor}

\newtheorem{theorem}{Theorem}
\newtheorem{lemma}{Lemma} 
\newtheorem{definition}{Definition} 

\newtheorem{proposition}{Proposition} 
\newtheorem{observation}{Observation}
\newtheorem{claim}{Claim}

\newenvironment{proof-sketch}{\noindent{\bf Proof Sketch:}\hspace*{1em}}{\qed}

\long\def\greybox#1{%
    \newbox\contentbox%
    \newbox\bkgdbox%
    \setbox\contentbox\hbox to \hsize{%
        \vtop{
            \kern\columnsep
            \hbox to \hsize{%
                \kern\columnsep%
                \advance\hsize by -2\columnsep%
                \setlength{1.1 \textwidth}{\hsize}%
                \vbox{
                    \parskip=\baselineskip
                    \parindent=0bp
                    #1
                }%
                \kern\columnsep%
            }%
            \kern\columnsep%
        }%
    }%
    \setbox\bkgdbox\vbox{
        \pdfliteral{0.85 0.85 0.85 rg}
        \hrule width  \wd\contentbox %
               height \ht\contentbox %
               depth  \dp\contentbox
        \pdfliteral{0 0 0 rg}
    }%
    \wd\bkgdbox=0bp%
    \vbox{\hbox to \hsize{\box\bkgdbox\box\contentbox}}%
    \vskip\baselineskip%
}

\begin{document}

\title{\textbf{Stochastic Stability in Schelling's Segregation Model with Markovian Asynchronous Update}
}

\author{Gabriel Istrate\footnote{Department of Computer Science, West University of Timi\c{s}oara, and e-Austria Research Institute, Bd. V. P\^{a}rvan 4, cam. 045 B, Timi\c{s}oara,  Romaniam email: gabrielistrate@acm.org. This work was supported by a grant of the Ministry of Research and Innovation, CNCS - UEFISCDI, project number PN-III-P4-ID-PCE-2016-0842, within PNCDI III.}}

\maketitle


\begin{abstract}
We investigate the dependence of steady-state properties of
Schelling's segregation model on the
agents' activation order. Our basic formalism is the Pollicott-Weiss
version  of Schelling's segregation model. Our
main result modifies this baseline scenario by incorporating {\it contagion} in the decision to move: (pairs of) agents are connected by a second, {\it agent  influence} network. Pair activation is specified
by a random walk on this network.

The considered schedulers choose the next pair \textit{nonadaptively}. We can complement this result by an example of adaptive scheduler (even one that is quite fair) that is able to preclude maximal segregation. Thus scheduler nonadaptiveness seems to be required for the validity of the original result under arbitrary asynchronous scheduling. The analysis (and 
our result) are part of an {\it adversarial scheduling approach} we are advocating to evolutionary games and social simulations.
\end{abstract}

\section{Introduction}
\label{intro}

Schelling's Segregation Model \cite{schelling-segregation} is one of
the fundamental dynamical systems of  Agent-Based Computational
Economics, perhaps one of the most convincing examples of Asynchronous Cellular Automata (ACA) \cite{fates2013guided} employed in the social sciences. It exhibits large-scale self-organizing neighborhoods, due to agents' desire to live close to their own kind. A remarkable feature
of the model that has captured the attention of social scientists is the fact that segregation is an {\it emergent} phenomenon, 
that may appear even in the presence of just mild preferences (at the individual level) towards living with one's own kind. The model has sparked a significant interest and work, coming from various areas such as Statistical Physics \cite{physical-schelling}, agent-based computational economics \cite{pancs-vriend,segregation-jebo-zhang}, game
theory \cite{peyton-young-book}, theoretical computer science \cite{brandt2012analysis,montanari2009convergence,auletta2013logit}, or applied mathematics \cite{pollicott-weiss}.

Schelling's segregation model is an asynchronous dynamical system on a graph (usually a finite portion of the one-dimensional or the two-dimensional lattice). It can be described, informally as follows: vertices in the graph are in one of three states: \textit{unoccupied}, when no agent sits on the given node, or one of \textit{red/blue} ($\pm 1$), corresponding to the color of the agent inhabiting the node. Agents have a (non-strict) preference towards living among agents of the same color. This is modeled by considering a \textit{local neighborhood} around the agent. Depending on the density of like-colored agents in the neighborhood the agent may be 
in one of two states: \textit{happy} and \textit{unhappy}. An unhappy agent may seek to trade places with another agent in order to become happy. It was originally observed via "pen-and-paper simulations'', and proved rigorously in a variety of settings, that segregated states may arise even when agent only have a weak preference for its own color, and are happy to live in a mixed neighborhood, as long as it contains ``enough'' of its own kind. Difference in the topology, activation order, specification of the update mechanism account for the dizzying variety of variants of the model that have been investigated so far (for the intellectual context of the model and a related one, due to Sakoda, see \cite{hegselmann2017thomas}). 

Qualitative properties of asynchronous cellular automata are highly dependent on activation order \cite{fates2004experimental,boure2012probing}. In particular, when viewed as dynamical systems, ACA  may exhibit a multitude of limit cycles
, and the update dynamics  "chooses" one of these limit cycle in a path dependent manner. The challenge then becomes to explain the selection of one particular limit cycle among many possible ones. 

 One particularly interesting class of techniques, brought to evolutionary games by Foster and Young \cite{foster1990stochastic} (see also \cite{peyton-young-book}) uses the  concept of \textit{stochastic stability} to deal with this problem.  It was the fundamental insight of Peyton-Young \cite{peyton-young-book} that adding continuous small perturbations to a certain dynamics might help ``steer''| the system towards a particular subset of equilibria, the so-called {\it stochastically stable states}. Indeed, in several versions of Schelling's segregation model
  \cite{young2001dynamics,peyton-young-book,segregation-jebo-zhang} the most segregates states are identified as precisely the stochastically stable equilibria of the dynamics. 

Though such results are interesting, they are still not realistic enough enough: results about stochastic stability in models on graphs may be sensitive to the precise specification of the update order, which in realistic scenarios need not be the random one. As noted in many papers, precise specification of an asynchronous 
schedule in social systems can arise from many factors, including geography or agent incentives \cite{page1997incentives}. It is thus important to study validity of baseline results under different scheduling models. A dramatic example of this type is that of the related model of \textit{logit response dynamics}, another model analyzed via stochastic stability \cite{kandorilearning}.  Going in this model from a random single-node update to parallel modes (the so-called \textit{revision process} of \cite{alos2010logit}) may lead (in general games) to the selection (via stochastic stability) of states that are not even Nash equilbria. In contrast, for local interaction games the parallel \textit{all-logit} rule has a Gibbs limiting distribution
\cite{auletta2013logit}, similar to the random update case. 

Neither random scheduling nor parallel update can accurately model {\it social contagion} phenomena, i.e., agents becoming active as a result of other agents' action, via communication or imitation. 
Thus it is of interest to study the robustness social models to variations in the update rule. Indeed, in  
\cite{adversarial-cie,adversarial-mscs} we have proposed an 
{\it adversarial approach to social simulations}.  Roughly speaking, this means that we consider the baseline  dynamics under random scheduling, then modify the update order to arbitrary scheduling, and attempt to derive necessary and/or sufficient conditions on the scheduler that make the results from the random update case extend to the adversarial setting. We have accomplished this in \cite{adversarial-mscs} for Prisoners' Dilemma with Pavlov strategy, a Markov chain previously investigated in \cite{ipd:colearning}, and in \cite{adversarial-cie} for the logit response dynamics. 

The purpose of this paper is to introduce contagion in 
evolutionary versions of Schelling's segregation model,  
as studied by Pollicott and Weiss \cite{pollicott-weiss}, and study the setting where the set of agents that becomes active is specified by a random walk on a second "communication" network. A similar model was investigated for the logit response dynamics in \cite{adversarial-cie}, and is apparently consistent with some real-life contagion phenomena in power networks \cite{hines2017cascading}. In the most general setting this communication network works on \textit{pairs of vertices}. The more natural case where agents influence each other is a special case of our setting. The feature of the Pollicott-Weiss model that is of special interest to our
study is that, although Schelling's model might have multiple equlibria, it is only the most segregated states that are stochastically stable.  Our result
shows that this extends to a scenario with social contagion: we prove a result with a similar flavor under a more general nonadaptive model of activation. 

Even though we use analytic rather than experimental techniques, our results are naturally related to a long line of research that investigates the robustness of discrete models under various scheduling models \cite{fates2004perturbing,boure2012probing}. On the other hand the notion we consider, that of stochastic stability, is highly related to the analysis of cellular automata using  dynamical systems techniques \cite{dennunzio2012cellular}. In contrast to many such studies, though, that only perturb the initial system state, stochastic stability embodies the notion of stability under continuous (but vanishingly small) perturbations. 
\vspace{-3mm}
\section{Preliminaries}
\vspace{-2mm}
We first review the notion of {\it stochastic stability} for perturbed dynamical systems described by Markov chains: 

\begin{definition}\label{def:perturbed-markov}
 Let the Markov chain $P^{0}$ be defined on a finite state set $\Omega$.
For every $\epsilon > 0$, we also define a Markov chain $P^{\epsilon}$ on $\Omega$.
Family $(P^{\epsilon})_{\epsilon \geq 0} $ is called a {\it regular perturbed Markov process} if all
of the following conditions hold: 
\begin{itemize}
\item For every $\epsilon > 0$ Markov chain $P^{\epsilon}$ is irreducible and aperiodic.
\item For each pair of states $x,y\in \Omega$, $
\lim_{\epsilon > 0} P_{xy}^{\epsilon} = P_{xy}^{0}$.
\item Whenever $P_{xy}=0$ there exists a real number $r(m)>0$, called the {\it resistance of
transition $m=(x\rightarrow y)$}, such that as $\epsilon \rightarrow 0$,
$P_{xy}^{\epsilon}= \Theta(\epsilon^{r(m)})$.
\end{itemize}

\noindent Let $\mu^{\epsilon}$ be the stationary distrib. of
$P^{\epsilon}$. State $s$ is {\it stochastically stable } if
${\varliminf}_{\epsilon \rightarrow 0} \mu^{\epsilon}(s) > 0$.
\end{definition}

We use a standard tool in this area: a result due to 
Young (Lemma 3.2 in \cite{peyton-young-book}) that allows us to recognize stochastically stable states in a Markov Chain using {\it spanning
trees of minimal resistance}:

\begin{definition}
If $j\in S(G)$ is a state, a {\it tree rooted at node $j$} is a set
$T$ of edges so that for any state $w\neq j$ there exists an unique
(directed) path from $w$ to $j$. The {\it resistance of a rooted
tree $T$} is defined as the sum of resistances of all edges in $T$.
\end{definition}

\begin{proposition}\label{characterization} {\bf (Young)}
The stochastically stable states of a regular Markov process $(P_{\epsilon})$ are precisely those states
$z\in \Omega$ such that there exists a tree $T$ rooted at $z$ of minimal resistance
(among all rooted trees).
\end{proposition}

\vspace{-3mm}
\section{The model}
\vspace{-2mm}

We consider an $N\times N$ two-dimensional lattice graph $G$ with periodic boundary conditions (that is, a torus). Let $V$ be the set of vertices of
this graph. Each vertex of $G$ hosts an agent, colored either red or blue. 
The neighborhood of a vertex $v$  is the {\it four-point neighborhood}, consisting of the cell to the left,up,right,down of the cell holding $v$.  An agent's utility is written as  
$\forall i\in V$, $u_{i}(x)=r\cdot w(x_{i})+\epsilon_{i},$ where $r$ is a positive constant, assumed similarly to  \cite{segregation-jebo-zhang} to be the same for all agents, and $w(x)$ is
defined, similarly to 
\cite{pollicott-weiss}, as the difference between the number of
neighbors of $x$ having the same color and the number of neighbors
of $x$ having the opposite color. Finally,  $\epsilon_{i}$ are (possibly different) agent-specific constants.




Next we specify our scheduling model, defined as follows:

\begin{definition}\label{mc}
 {\bf [Markovian contagion]:} To each pair of vertices $e$ we associate a probability distribution $D_{e}$ on $V\times V$ such that $e\in supp(D_{e})$\mbox{  }\footnote{This translates, intuitively, to the following condition: we always give the participants in a swap the chance to immediately reevaluate their last move.}. 
We then choose the pair to be scheduled next as follows: Let $p_{i}$ be the pair chosen 
at stage $i$. Select the next scheduled pair $p_{i+1}$  by
sampling from the set of pairs in $D_{p_{i}}$. We assume that for any two pairs $e,e^{\prime}$ the following condition holds: 
\begin{equation}
Pr[e\rightarrow e^{\prime}]>0\Leftrightarrow Pr[e^{\prime} \rightarrow e] > 0. 
\label{revers}
\end{equation}

\end{definition}

In other words: the next scheduled pair only depends on the last
scheduled pair, succession relation $e\rightarrow e^{\prime}$ specifies a \textit{bidirected graph}
$H(G)$ whose vertex set is $V\times V$, and the scheduled pair can be seen as
performing a random walk (possibly a non-uniform one) on $H(G)$. In particular the next chosen pair is {\it not}
guaranteed to have different labels on endpoints. Furthermore, graph $H(G)$ should be connected, otherwise the choice of a particular initial sequence of moves could preclude a given edge from ever being scheduled sometimes in the future.

\begin{observation}\label{cont}
 A particular case of Definition~\ref{mc}, which justifies the name
contagion is described informally as follows: {\it agents}, rather
than pairs, are given the opportunity to switch. They randomly choose
a swapping neighbor among those available to them. There exists a
second, separate {\it influence  network} $I$. The next scheduled
agent is one of the neighbors (in $I$) of the previously scheduled
agent. Indeed, to describe this scenario in the setting of the previous definition, define $D_{e}$ to consist of pairs $e^{\prime}$ that share with $e$ a vertex.  
\end{observation}

\begin{observation}
 The random scheduler is a particular case of Definition~\ref{mc}, when $D_{e}$ is the uniform distribution on $V\times V$. 
\end{observation}

To complete the description of the dynamics, we only need to specify
the probability that two agents inhabiting the different endpoints
of a pair $e=(u,v)$ switch when pair $e$ is scheduled. This is
accomplished using the so-called \textit{log-linear response rule} \cite{blume-population-games,peyton-young-book,alos2010logit},
specified as follows: let $S$ be the state before the switch and $T$ be the state obtained if the two agents at $u,v$ switch. Then: 

\begin{equation}
 Pr[S\rightarrow T]=\frac{e^{\beta\cdot [u_{1}(T)+
 u_{2}(T)]}}{e^{\beta[u_{1}(S)+u_{2}(S)}]+e^{\beta[u_{1}(T)+u_{2}(T)]}},
 \label{eq:def} 
\end{equation}

where $u_{1},u_{2}$ are the corresponding utility functions of the
two agents at the endpoints of the scheduled pair, and $\beta >0$ is
a constant. This is, of course, the noisy version of the
best-response move, that would choose the move that maximizes the
sum of utilities $u_{1}(\cdot)+u_{2}(\cdot)$.

The \textit{state of the system} is defined by a  vector $\overline{w}\in \{-1,+1\}^{V}$ is a vector encoding the labels of all vertices of the torus. To obtain a description of the dynamics as an aperiodic Markov chain we have to complete the description of the system state by a pair $r$ of vertices, i.e., the last pair that had the opportunity to switch by being scheduled.  Thus the state space of the Markov chain is $S(G):= {\{\pm 1 \}^V} \times (V\times V)$. When $H(G)$ is strongly connected, $M_\beta$ is ergodic, so it has an unique stationary distribution
$\Pi_\beta$.  It is easy to see (and similar to previous results e.g. in \cite{peyton-young-book}) that the family of chains $(M_\epsilon)_{\epsilon > 0}$ is a regularly perturbed Markov process (where we define $\epsilon = e^{-\beta}$). 
\begin{definition}
 A state $w\in {\{\pm 1 \}^V} $ is called {\it maximally segregated} if $w$ realizes the minimum value of the number of red-blue {\it edges} (of the torus), across all possible states on $G$ with a given number of red/blue agents.
 \end{definition}

In \cite{pollicott-weiss} a complete characterization of a maximally
segregated state was obtained (Theorem 2 in that paper). Roughly
they are horizontal or vertical "bands", possibly with a "strip"
attached, or a rectangle, possibly with at most two "strips"
attached. We refer the reader to 
\cite{pollicott-weiss} for details, and don't discuss it any further.
\vspace{-3mm}
\section{Main result and its interpretation.}
\vspace{-2mm}
Our main result is: 

\begin{theorem}\label{young-contagion} \label{rwalk} The stochastically stable states for
Schelling's segregation model with Markovian contagion form a subset
of the set $Q\subseteq S(G)$, 
\begin{equation} 
Q=\{(w,e)|w\mbox{ is maximally segregated and }e\in V\times V\}
\end{equation} 
\end{theorem}

\noindent In other words: \textbf{the conclusion that stochastically stable states in Schelling's segregation model are  maximally segregated  is robust to extending the update model from a random one to those from the family from Definition~\ref{mc}, that incorporate Markovian contagion.}

The defining  feature of the class of schedulers in the previous result seems to be that they choose the next scheduled pair \textit{nonadaptively}: the next pair only depends on the last scheduled pair, and \textbf{not} on other particulars of the system state. Indeed, in the full version of the paper we will complement the result above by another one (very easy to state and prove), that shows that  some \textbf{adaptive} scheduler (despite being quite fair) can forever preclude the system from ever reaching maximal segregation (thus "breaking the baseline stylized result").   
\begin{proof} 

We will employ a fundamental property, noted for models of segregation such as the one in this paper e.g. in \cite{segregation-jebo-zhang}: they are \textit{potential games} \cite{potential-games}, i.e., they admit a function $L:S(G)\rightarrow {\bf R}$ such
that, for any player $i$, any strategy profile (i.e., vector of player strategies) $x_{-i}:= (x_{j})_{j\in V, j\neq i}$, and any two strategies $z,t$ for player $i$
\begin{equation} 
u_{i}(z;x_{-i})-u_{i}(t;x_{-i}) = L(z;x_{-i})- L(t;x_{-i}).
\end{equation} 

In other words, differences in utility of the $i$'th player as a result of using different strategies are equal to the differences in potential among the two corresponding profiles. The function $L$ is defined simply
as $L(s)=\sum_{i} u_{i}(s).$ Strictly speaking the potential above is defined for the original Policott-Weiss model i.e., defined on $\{\pm 1\}^{V}$, instead of $S(G)$. But it can be easily extended by simply applying it to any pair $(s,e)$ (thus neglecting $e$).
Moreover, the following property holds, which determines the resistance of moves: 

\begin{lemma} Let $A\in\{\pm 1\}^V$ be a state of the system, and let $e=(i,j)$ be a pair in $V\times V$. Let $B$ be a state obtained by making the move $m=A\rightarrow B$
($B$ is either $A$ or is the state obtained from $A$ by swapping the states $A_i,A_j$). Then the resistance $r(m)$ of move $m$ is equal to:
\begin{equation} 
r(m)=[u_{i}(A)+u_{j}(A)]- [u_{i}(B)+u_{j}(B)]=2(L(A)-L(B)) >0
\end{equation} 
 when $A\rightarrow B$ is a swap that diminishes potential, and to 
\begin{equation} 
[u_{i}(C)+u_{j}(C)]- [u_{i}(A)+u_{j}(A)]=2(L(C)-L(A)) >0
\end{equation} 
 when $A=B$, but the corresponding swap $A\rightarrow C$ would be a potential improving one.  In all other cases $r(m)=0$ (that is $m$ is a neutral move). 
\label{diff}
\end{lemma}

In other words, a move has positive resistance when one of the
following two alternatives hold: (a). The move corresponds to a
decrease in potential. The resistance of the move is, in this case, equal to the potential decrease. 
(b). The move corresponds to preserving the current state (as well as agents' utilities), but the other possible move would have led to a state of higher potential. The resistance of the move is, in this case,  equal to the difference in potentials between this better state and the current one. 

\begin{proof} 
Follows directly from equation~(\ref{eq:def}) and the definition of resistance. 
\end{proof} 

We apply this result to prove the following lemma: 

\begin{lemma} Consider a state $Y\in Q$ that is maximally segregated. Consider another state $X$, and a tree $T$ rooted at $X$ having minimal potential. Then there exists another tree
$\overline{T}$ rooted at $Y$ whose potential is at most that of tree $T$, strictly less in case when $X$ is not a
maximally segregated state.
\end{lemma}

\begin{proof} 

Note that, by the definition of utility functions, maximally segregated states are those that maximize the potential. 

Since $T$ is an oriented tree, there is an unique directed path 
\begin{equation} 
p:[Y=(s_{0},e_{0})\rightarrow  \ldots
\rightarrow (s_{k},e_{k})\rightarrow (s_{k+1},e_{k+1})\rightarrow \ldots \rightarrow
(s_{r},e_{r})=X]
\label{path-p}
\end{equation} 

in $T$ from $Y$ to $X$. Here $s_{0},s_{r}\in \{\pm 1\}^{V}$ are states, and $e_{0},e_{r}$ are pairs in $V\times V$. First, we decompose $T$ into three subsets as follows: 

\begin{enumerate}
\item The set of edges of $p$ (see Figure~\ref{tree-fig})
\item The set of edges of, $W_{Y}$ the subtree rooted at $y$.
\item The edges of subtrees $W_{k}$ rooted at nodes $(s_{k},e_{k})$ of $p$, other than $Y$ (but possibly including $X$).
\end{enumerate}

\begin{figure} 
\begin{center} 
\begin{tikzpicture}
\GraphInit[vstyle=Normal]
\Vertex[x=-1,y=-2]{Y}
\Vertex[NoLabel,x=1,y=-2]{Z}
\Vertex[NoLabel,x=3,y=-2]{R}
\Vertex[NoLabel,x=5,y=-2]{P}
\Vertex[x=7,y=-2]{X}
\Vertex[NoLabel,x=2,y=-4]{A}
\Vertex[NoLabel,x=4,y=-4]{B}
\Vertex[NoLabel,x=-2,y=-4]{C}
\Vertex[NoLabel,x=0,y=-4]{D}
\Vertex[NoLabel,x=6,y=-4]{E}
\Vertex[NoLabel,x=8,y=-4]{F}
\Edges[style=->,label=type 1](Y,Z)
\Edges[style=->,label=$\cdots$](Z,R,P)
\Edges[style=->,label=type 1](P,X)
\Edges[style=->,label=type 2](C,Y)
\Edges(C,D)
\Edges[style=->,label=type 2](D,Y) 
\Edges[style=->,label=type 3](E,X)
\Edges(E,F)
\Edges[style=->,label=type 3](F,X)
\Edges[label=type 3](R,A)
\Edges[label=$W_{k}$](A,B)
\Edges[label=type 3](B,R)
\end{tikzpicture}
\end{center}
\caption{Decomposition of edges of tree $T$. Path $p$ is on top.}
\label{tree-fig}
\end{figure}
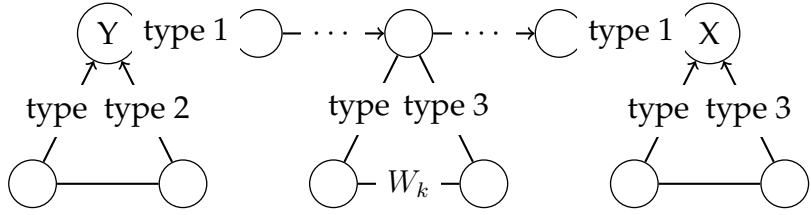 

\begin{lemma} 
Without loss of generality we may assume that the path $p$ contains no two consecutive vertices $(s_{k},e_{k})$ and $(s_{r},e_{r})$, $k<r$, with $s_{k}=\ldots = s_{r}$ and all the moves between $s_{k}$ and $s_{r}$ having zero resistance. 
\label{foo}
\end{lemma} 
\begin{proof}
Suppose there was such a pair $k,r$. Take one with maximal $k-r$. Define a tree $T^{\prime}$ by
\begin{itemize} 
\item First connecting $(s_{k},e_{k})$ directly to $(s_{r+1},e_{r+1})$. This is legal since $s_{k}=s_{r}$. Indeed, since activating edge $e_{r+1}$ move the system from $s_{r}$ to $s_{r+1}$, activating the same edge moves the system from $(s_{k},e_{k})$ to $s_{r+1}$ as well. 
\item Also connecting $(s_{k+1},e_{k+1}),\ldots, s_{r},e_{r})$ directly to $(s_{r+1},e_{r+1})$
\end{itemize} 
It is easy to see that $T^{\prime}$ is a tree with the same resistance as $T$, since the only removed edges have zero resistance. 
\end{proof}\qed

In particular, applying iteratively Lemma~\ref{foo}, we may assume that all transitions on the path from $p$ to $q$ either change the system state or have positive resistance.

Now, to obtain tree $\overline{T}$ we will first obtain a graph $\overline{\overline{T}}$, in which every node has \textit{at least one path} to node $Y$. "Thinning out" this directed graph to a tree yields a tree $\overline{T}$ of even lower resistance. To obtain directed graph $\overline{\overline{T}}$: 

\begin{enumerate}
 \item First, add to $\overline{\overline{T}}$ the edges of $T_{Y}$.

\item Next, define path $q$ from $X$ to $Y$ as follows: 
\begin{equation} 
 q:[(s_{r},e_{r})\rightarrow
(s_{r-1},e_{r})\rightarrow (s_{r-2},e_{r-1}) \rightarrow \ldots \rightarrow
(s_{0},e_{1}) \rightarrow (s_{0},e_{0})]. 
\end{equation}

In other words, $q$ aims to ``undo'' the sequence of moves in path $p$ from equation~(\ref{path-p}) from $Y$ to $X$. However, since the states of the Markov chain also have as a second component a pair in $V\times V$, (corresponding to the last scheduled edge), we need to take a little extra care when defining $q$. Specifically $q$ starts at $X$ but \textbf{cannot} simply reverse the edges of $p$, since these do \textbf{not} correspond to legal moves. To define $q$, we first make a move at $e_{r}$ by ``undoing'' the last move of $p$\mbox{ } \footnote{this is where we use a property specific to our model of Schelling segregation, as opposed to proving a result valid for general potential game: the property that we employ is that in Schelling's model any move $m$ "can be undone". This means that there is a move $n$ using the same pair of vertices as $m$ that brings the system back to where it was before. Move $n$ simply "swaps back" the two agents if they were swapped by $m$, and leaves them in place otherwise.}. This yields state $(S_{r-1},e_{r})$ (since pair $e_{r}$ is scheduled in this as well). We then continue to ``undo moves of $p$'' until the state becomes $S_{0}$. This is possible because of condition~(\ref{revers}): since pair $e_{r+1}$ can be scheduled after $e_{r}$, scheduling $e_{r}$ can move the system from $s_{r}$ to $s_{r-1}$, scheduling pair $e_{r-1}$ can move the system from $s_{r-1}$ to $s_{r-2}$, and so on, until state $S_{0}$ is reached.  
At this moment the last activated pair was $e_{1}$. $q$ then moves to $Y$ by making a move (with no effect) on $e_{0}$. 

Note that every such  path will contain, for every $k=1,r$, a vertex whose state is $s_{k}$. 

\item For every tree component $W_{k}$ obtained by removing path $p$ from $T$, attached to $p$ at $(s_{k},e_{k})$ perform one of the following: 
\begin{enumerate} 
\item[-]{\bf Case 1:} {\it $s_{k}=s_{k-1}$}.

In this case the point $(s_{k},e_{k})=
(s_{k-1},e_{k})$ is on path $q$ as well, therefore we also add the rooted tree $W_{k}$ to $\overline{\overline{T}}$. This is possible since attaching a tree to a node depends only on the system state, but \textit{not} on the last scheduled node. Moreover, \textit{the resistance of $W_{k}$ does not change} as a result of this attachment. 

\item[-]{\bf Case 2:} {\it $s_{k}\neq s_{k-1}$ and move
$(s_{k-1},e_{k-1})\rightarrow (s_{k},e_{k})$ has resistance $>0$.}

In this case, since in configuration $s_{k-1}$ and scheduled move
$e_{k}$ we have a choice between moving to $s_{k}$ and staying in
$s_{k-1}$, it follows that the transition $(s_{k-1},e_{k-1})\rightarrow
(s_{k-1},e_{k})$  has zero resistance and $L(s_{k})<L(s_{k-1})$. 
Hence transition $(s_{k},e_{k})\rightarrow (s_{k-1},e_{k})$ has zero resistance. 


We now add  the tree $\overline{W_{k}}=W_{k}\cup 
\{(s_{k},e_{k})\rightarrow (s_{k-1},e_{k})\}$ (rooted at node $(s_{k-1},e_{k})$, which is on $q$) to  $\overline{\overline{T}}$.   The tree $\overline{W_{k}}$ has the same total resistance as $W_{k}$. All nodes from $W_{k}$, including $(s_{k},e_{k})$ can now reach $Y$ via $q$. 

\item[-]{\bf Case 3:} \textit{$s_{k-1}\neq s_{k}$, move
$(s_{k-1},e_{k-1})\rightarrow (s_{k},e_{k})$ has zero resistance and all moves on $p$ between $s_{k}$ and $X$ have zero resistance.}

Then we add to $\overline{\overline{T}}$ this portion of $p$, together with $W_{k}$. This way we connect nodes in $W_{k-1},W_{k}$ to $Y$ (via $X$ and $q$). All added edges except those of one of the trees $W_{l}$ have zero resistance. 

\item[-]{\bf Case 4:} \textit{$s_{k-1}\neq s_{k}$, the move
$(s_{k-1},e_{k-1})\rightarrow (s_{k},e_{k})$ has zero resistance, but some 
move on $p$, between $s_{k}$ and $X$ has positive resistance.}

Let $(s_{k+l},e_{l})\rightarrow (s_{k+l+1},e_{l+1})$ be the closest move (i.e., the one that minimizes $l$) with positive resistance. 

If $s_{k+l}=s_{k+l+1}$ then we have already connected $(s_{k+l},e_{l})$ to $Y$, as it falls under Case 1. Now just add all the (zero resistance) edges of $p$ between $s_{k},e_{k}$ 
and $(s_{k+l},e_{k+l})$, together with edges of $W_{k}$, to connect all such nodes to $Y$. 

If $s_{k+l}=s_{k+l+1}$ then we have already connected $(s_{k+l},e_{l})$ to $Y$, as it falls under Case 2. We proceed similarly. 
\end{enumerate} 
\end{enumerate}

The previous construction has ensured that any pair $(s,e)$ is connected by \textit{at least} one path to $Y$.
Thinning out $\overline{\overline{T}}$ we get a rooted tree $\overline{T}$ having resistance less or equal to the resistance of $\overline{\overline{T}}$.  Since the four outlined transformations only add, in addition to trees $W_{k}$, edges of zero resistance, to compare the
total resistances of $T$ and $\overline{T}$ one should simply compare the total resistances of paths $p$ and
$q$. We claim that this difference in resistances of these paths is equal to the difference in 
potentials: 
\begin{claim} 
$r(p)-r(q)=2(L(Y)-L(X))\geq 0.$
\label{last}  
\end{claim}  
Proving Claim~(\ref{last}) would validate our conclusion, since $L(Y)-L(X)\geq 0$, and $L(X)=L(Y)$ iff $X$ is a global minimum state for the potential function.  We prove this by considering the correspondence between edges of paths $p$ and $q$: to each edge $e$ of $p$ one can associate an unique edge $e^{\prime}$ of $Q$ that ``undoes e''.

By the additivity of both resistance and potential, it is enough to prove that, for every edge $e$ of $p$ and its associated edge of $q$, $e^\prime$, $r(e^{\prime})-r(e)$ is equal to twice the difference in potentials between $S_{fin}$, the final state for the forward transition and $S_{init}$, the initial state.  The first thing to note is none of the two resistances can be infinite: the transition $e\rightarrow e^{\prime}$ corresponds to a move of the perturbed Markov chain (optimal or not). Its inverse corresponds to "undoing" that move, which is a legal move (eventually perturbed) in itself.  We employ Lemma~\ref{diff} and identify several cases: 
\begin{itemize}
\item \textit{The move  $e$ corresponds to \textbf{not} switching, and its resistance is zero.} Let $S_{1}$ be the common state, and let $S_{2}$ be the state corresponding to a switch. 
Then $S_{1}=S_{init}=S_{fin}$. Also, $L(S_{1}) \geq L(S_{2})$. So the move $e^{\prime}$ also stays in state $S_{1}$ (when it could have gone to $S_{2}$), which is the optimal action, given that $L(S_{1}) \geq L(S_{2})$. Thus in this case both the "forward" and the "backward" transition have resistance zero, and do not count towards the sum of resistances on the path. 
\item \textit{The move  $e$ corresponds to \textbf{switching}, and its resistance is zero.} Then, $S_{fin}=S_{2}$, $S_{init}=S_{1}$. By Lemma~\ref{diff} $L(S_{2})>L(S_{1})$.  The backward move has positive resistance equal to $2(L(S_{2})-L(S_{1}))$. The result is verified. 

\item \textit{The move  $e$ corresponds to \textbf{not} switching, and its resistance is nonzero.}  Then $S_{init}=S_{fin}=S_{1}$ and $L(S_{1})<L(S_{2})$ (since switching would be beneficial). Therefore in the backward move the state stays $S_{1}$ (when it could have gone to $S_{2}$). The resistance is equal to $2(L(S_{2})-L(S_{1}))$, the same as the resistance of the forward move. Therefore $r(e^{\prime})-r(e)=L(S_{fin})-L(S_{init})=0$. 
\item \textit{The move  $e$ corresponds to \textbf{switching}, and its resistance is nonzero.} Then, $S_{fin}=S_{2}$, $S_{init}=S_{1}$, by Lemma~\ref{diff} $L(S_{2})<L(S_{1})$ and the resistance of the forward move is equal to $2(L(S_{1})-L(S_{2}))=2(L(S_{init})-L(S_{fin}))$. The resistance of the backward move is equal to zero, so the result is verified in this case as well. 

Thus the claim is established and the proof of the theorem is complete. 
\end{itemize} 

\end{proof} 
\vspace{-3mm}
\end{proof}
\section{Outlook and Further Work}

Theorem~\ref{young-contagion} is only the main result in the
adversarial analysis of Schelling's segregation model. It shows that
stochastically segregated states are maximally segregated. Is the
converse true ? Namely, is every maximally segregated state
stochastically stable ? Such a result is indeed true in 1D versions of Schelling's segregation model (such as the one presented in \cite{peyton-young-book}). We will discuss the  2D case with Markovian contagion in 
the journal version of the paper. 

Other topics deserving research include studying conditions that preclude segregation, determining the convergence time of the segregation dynamics with Markovian contagion, models with Markovian contagion and concurrent updates \cite{montanari2009convergence,auletta2013logit}, etc. We plan to address these and other issues in follow-up papers. 
\vspace{-3mm}
\small{

}
\end{document}